%% file: zhanglin-QuantumFidelity10.tex
\documentclass[11pt]{article}

\input{preamble.tex}

\begin{document}

\title{\Large Quantum fidelity and relative entropy between unitary orbits}

\author{Lin Zhang\footnote{E-mail: godyalin@163.com}\\[1mm]
  {\it\small Institute of Mathematics, Hangzhou Dianzi University, Hangzhou 310018, PR~China}\\
  Shao-Ming Fei\\[1mm]
  {\it\small School of Mathematics of Sciences, Capital Normal University, Beijing 100048, PR China}\\[1mm]
  {\it\small Max-Planck-Institute for Mathematics in the Sciences, Leipzig 04103, Germany}}

\date{}
\maketitle \mbox{}\hrule\mbox\\
\begin{abstract}

Fidelity and relative entropy are two significant quantities in
quantum information theory. We study the quantum fidelity and
relative entropy under unitary orbits. The maximal and minimal
quantum fidelity and relative entropy between two unitary orbits are
explicitly derived. The potential applications in quantum
computation and information processing are discussed.

\end{abstract}
\mbox{}\hrule\mbox

\section{Introduction}

There are many important quantities in characterizing a bipartite or
multipartite quantum state, such as the mutual information, quantum
correlation and entanglement etc. in quantum information theory. To
investigate the variations of such quantities under \emph{unitary
dynamics} has practical applications. In \cite{San,MU}, by searching
for the maximally and minimally correlated states on a unitary
orbit, the authors studied the amount of correlations, quantified by
the quantum mutual information, attainable between the components of
a quantum system, when the system undergoes isolated, unitary
dynamics. The correlations in a bipartite or multipartite state
within the construction of unitary orbits have been also examined in
\cite{MG}.

In fact, there are also many important quantities in characterizing
the relations between two quantum states, such as the quantum
fidelity and relative entropy. They give rise to the measures of a
kind of distance between two quantum states. They can be also used
to characterize the property of a given quantum state, for instance,
to quantify the quantum entanglement between two parts of a state,
which is the shortest distance between the state and the set of all
separable states. Such distances between two quantum states have
many applications in quantum information processing. In \cite{Karol}
it has been shown that the problem of deterministically quantum
state discrimination is equivalent to that of embedding a simplex of
points whose distances are maximal with respect to the Bures
distance or trace distance of two quantum states.

In this paper, we study the quantum fidelity and relative entropy
under arbitrary unitary dynamics. Under general unitary evolutions,
every given quantum state belongs to a continuous orbit. We analyze
the `distance' between two quantum states under general unitary evolutions:
the maximal and minimal quantum fidelity and relative entropy
between two such unitary orbits, by using the
combinatory techniques in majorization theory and operator monotones.
It is also shown that they are intervals between these
minimal and the maximal values.

The paper is organized as follows: In Sect.~\ref{sect:fidelity} we
derive the maximal and minimal values of the quantum fidelity
between the unitary orbits of two quantum states. Moreover, we prove
that the values of the quantum fidelity fill out an interval. We
also discuss the fidelity evolution generated by Hamiltonian. In
Sect.~\ref{sect:relative-entropy} we consider the optimal problems
for relative entropy and derive the maximal and minimal values of
the relative entropy between the unitary orbits of two quantum
states. We summarize and discuss in Sect.~\ref{sect:conclusion}.

\section{Quantum fidelity between unitary orbits}\label{sect:fidelity}

The \emph{fidelity} between two $d\times d$ quantum states,
represented by density operators $\rho$ and $\sigma$, is defined as
\begin{eqnarray}
\rF(\rho, \sigma ) = \Tr{\sqrt{\sqrt{\rho}\sigma\sqrt{\rho}}} \equiv
\Tr{\abs{\sqrt{\rho}\sqrt{\sigma}}}.
\end{eqnarray}
This is an extremely fundamental and useful quantity in quantum
information theory. In quantum information processing, one wishes to
transform a given quantum state to the final target state.
For instance, in the quantum computation with qubits or qutrits,
it is essential to estimate the "distance" between the desired target state
and the approximate state that can be realized by projected Hamiltonian \cite{nielsen,lyf}.
Practically, due to the inevitable
interaction between the quantum systems and its environment and possible experimental imperfectness,
it is crucial to characterize quantitatively to what extent can an evolved quantum
state be close to the target state. For this purpose, the fidelity is often used as the measure of
the distance between two quantum states.
The squared fidelity above has been called \emph{transition probability} \cite{Uhlmann76,Uhlmann2011}.
Operationally it is the maximal success probability of transforming a state to another one
by measurements on a larger quantum system. The fidelity is also employed in a number of
problems such as quantifying entanglement \cite{Vedral} and quantum error correction \cite{Kosut}.

Let $\unitary{\cH_d}$ denote the set of $d\times d$ unitary matrices
on $d$-dimensional Hilbert space $\cH_d$. For a given density matrix
$\rho$, its \emph{unitary orbit} is defined by
\begin{eqnarray}
\cU_\rho = \Set{U\rho U^\dagger : U\in \unitary{\cH_d}}.
\end{eqnarray}
Clearly a density operator $\rho$ whose evolution is governed by a
von Neumann equation remains in a single orbit $\cU_\rho$. The
orbits $\cU_\rho$ are in one-to-one correspondence with the possible
spectra for density operators $\rho$.

We investigate the bound of the quantum fidelity between the unitary
orbits $\cU_\rho$ and $\cU_\sigma$ of two quantum states $\rho$ and
$\sigma$. Due to the unitary invariance of fidelity, the problem
boils down to determining the following extremes:
$\min_{U\in\unitary{\cH_d}} \rF(\rho, U\sigma U^\dagger)$ and
$\max_{U\in\unitary{\cH_d}} \rF(\rho, U\sigma U^\dagger)$.

Let $\rF(p,q)= \sum_j \sqrt{p_jq_j}$ denote the classical fidelity
between two probability distributions $p=\set{p_j}$ and
$q=\set{q_j}$.

\begin{thrm}\label{th:1}
The quantum fidelity between the unitary orbits $\cU_\rho$ and
$\cU_\sigma$ satisfy the following relations;
\begin{eqnarray}
\max_{U\in\unitary{\cH_d}} \rF(\rho, U\sigma U^\dagger)
&=& \rF(\lambda^\downarrow(\rho),\lambda^\downarrow(\sigma)),\label{t11}\\
\min_{U\in\unitary{\cH_d}} \rF(\rho, U\sigma U^\dagger) &=&
\rF(\lambda^\downarrow(\rho),\lambda^\uparrow(\sigma)),\label{t12}
\end{eqnarray}
where $\lambda^{\downarrow}(\rho)$ (resp. $\lambda^{\uparrow}(\rho)$
is the probability vector consisted of the eigenvalues of $\rho$,
listed in decreasing (resp. increasing) order.
\end{thrm}

\begin{proof}
We prove this Theorem for non-singular density matrices. The
general case follows by continuity. Indeed, assume that the theorem is correct
for non-singular density matrices.
Let $\sigma$ be singular. Then $\sigma+\varepsilon\I$ is non-singular.
Since $\lim_{\varepsilon\to0^+}\lambda^\downarrow(\sigma+\varepsilon\I) = \lambda^\downarrow(\sigma)$ and
$$
\max_{U\in\unitary{\cH_d}} \rF(\rho, U(\sigma+\varepsilon\I) U^\dagger)
= \rF(\lambda^\downarrow(\rho),\lambda^\downarrow(\sigma+\varepsilon\I)),
$$
by taking the limit $\varepsilon\to0^+$, we have that the theorem will be also true for singular density matrices.

Since the eigenvectors of two density matrices can always be
connected via a unitary, the problem is reduced to the case where
$[\rho,\sigma]=0$. Without loss of generality, we assume that $\rho$
and $\sigma$ have the following spectral decompositions:
$$
\rho = \sum^d_{j=1}
\lambda^\downarrow_j(\rho)\out{j}{j}\quad\text{and}\quad\sigma =
\sum^d_{j=1} \lambda^\downarrow_j(\sigma)\out{j}{j},
$$
where $\lambda^\downarrow_j(\rho)$ and
$\lambda^\downarrow_j(\sigma)$ are the eigenvalues of states $\rho$
and $\sigma$ respectively.

It has been shown that for any $n\times n$ Hermitian matrices $A$
and $B$, there exist two unitary matrices $V_1$ and $V_2$ such that
\cite{So},
$$
\exp{(A/2)}\exp{(B)}\exp{(A/2)}= \exp\Pa{V_1 A V^\dagger_1+V_2 B
V^\dagger_2}.
$$
Hence for Hermitian matrices $\rho$ and $U\sigma U^\dagger$, we have
$V_1$ and $V_2\in \unitary{\cH_d}$ such that
\begin{eqnarray}
\sqrt{\rho}U\sigma U^\dagger\sqrt{\rho} = \exp\Pa{V_1\log\rho
V^\dagger_1+V_2U\log\sigma U^\dagger V^\dagger_2}.
\end{eqnarray}
Therefore
\begin{eqnarray*}
&&\rF(\rho,U\sigma U^\dagger) = \Tr{\sqrt{\sqrt{\rho}U\sigma
U^\dagger\sqrt{\rho}}}\\
&&= \Tr{\exp\Pa{\frac{V_1\log\rho V^\dagger_1+V_2U\log\sigma
U^\dagger V^\dagger_2}2}}\\
&&= \Tr{\exp\Pa{\frac{\log\rho + \widehat U\log\sigma \widehat
U^\dagger}2}},
\end{eqnarray*}
where $\widehat U = V^\dagger_1V_2U$.

As for arbitrary Hermitian matrices $A$ and $B$, one has the
\emph{Golden-Thompson's inequality}:
$$
\Tr{e^{A+B}}\leqslant \Tr{e^A e^B},
$$
in which the equality holds if
and only if $[A,B]=0$ \cite{Soexp,Ruskai}, we have
\begin{eqnarray}\label{eq:G-T=}
\Tr{\exp\Pa{\frac{\log\rho + \widehat U\log\sigma \widehat
U^\dagger}2}} \leqslant \Tr{\sqrt{\rho}\widehat
U\sqrt{\sigma}\widehat U^\dagger}\leqslant \rF(\rho,\widehat
U\sigma\widehat U^\dagger).
\end{eqnarray}
Since the unitary group $\unitary{\cH_d}$ is compact, the
supremum is actually attained on some unitary. Let $U_0\in
\unitary{\cH_d}$ be such that
\begin{eqnarray*}
\max_U\rF(\rho,U\sigma U^\dagger)= \rF(\rho,U_0\sigma U^\dagger_0)=
\Tr{\exp\Pa{\frac{\log\rho + \widehat U_0\log\sigma \widehat
U^\dagger_0}2}}.
\end{eqnarray*}
We see that $\rF(\rho,U_0\sigma U^\dagger_0) = \rF(\rho,\widehat U_0\sigma
\widehat U^\dagger_0)$, namely, the inequality \eqref{eq:G-T=} must be an equality. Hence
$\Br{\rho,\widehat U_0\sigma \widehat U^\dagger_0}=0$, and $\widehat
U_0$ is just a permutation operator since $[\rho,\sigma]=0$.

We have shown that if $[\rho,\sigma]=0$, then there exists a
permutation matrix $P$ such that
$$
\max_U \rF(\rho,U\sigma U^\dagger) = \rF(\rho,P\sigma P^\dagger).
$$
Obviously the maximum is attained when the permutation $P$ is the
identity operator $\I_d$. That is, if $[\rho,\sigma]=0$, then
$$
\max_U\rF(\rho,U\sigma U^\dagger) = \rF(\rho,\sigma) = \sum^d_{j=1}
\sqrt{\lambda^\downarrow_j(\rho)\lambda^\downarrow_j(\sigma)},
$$
which proves \eqref{t11}.

On the other hand, we have
\begin{eqnarray*}
\rF(\rho,U\sigma U^\dagger) &=&
\Tr{\abs{\sqrt{\rho}U\sqrt{\sigma}U^\dagger}} \geqslant
\Tr{\sqrt{\rho}U\sqrt{\sigma}U^\dagger}.
\end{eqnarray*}
Since for Hermitian matrices $A$ and $B$ \cite{Bhatia},
\begin{eqnarray}\label{l33}
\inner{\lambda^\downarrow(A)}{\lambda^\uparrow(B)}\leqslant \Tr{AB}
\leqslant \inner{\lambda^\downarrow(A)}{\lambda^\downarrow(B)},
\end{eqnarray}
where $\inner{u}{v}:= \sum_j \bar u_j v_j$, we obtain
\begin{eqnarray}
\min_U \rF(\rho,U\sigma U^\dagger) &\geqslant&
\min_U\Tr{\sqrt{\rho}U\sqrt{\sigma}U^\dagger}=
\sum^d_{j=1}\sqrt{\lambda^\downarrow_j(\rho)\lambda^\uparrow_j(\sigma)}.
\end{eqnarray}
The above inequality becomes an equality for $U\in \unitary{\cH_d}$
such that $U\ket{j} = \ket{d-j+1}$, which proves \eqref{t12}.
\end{proof}

Theorem~\ref{th:1} gives a easy way to estimate the maximal and
minimal values of the quantum fidelity between the unitary orbits of
two quantum states. They are simply given by the eigenvalues of the
density matrices.

\begin{thrm}\label{th:2}
The set $\Set{\rF(\rho, U\sigma U^\dagger):  U\in \unitary{\cH_d}}$
is identical to the interval
\begin{eqnarray}
\Br{\rF(\lambda^\downarrow(\rho),\lambda^\uparrow(\sigma)),
\rF(\lambda^\downarrow(\rho),\lambda^\downarrow(\sigma))}.
\end{eqnarray}
\end{thrm}

\begin{proof} Note that any unitary matrix $U$ can be parameterized as
$U=\exp(tK)$ for some skew-Hermitian matrix $K$. In order to prove
the set $\Set{\rF(\rho, U\sigma U^\dagger):  U\in \unitary{\cH_d}}$
is an interval, we denote
\begin{eqnarray}
g(t) = \rF(\rho,U_t\sigma U^\dagger_t) =
\Tr{\sqrt{\sqrt{\rho}U_t\sigma U^\dagger_t\sqrt{\rho}}},
\end{eqnarray}
where $U_t = \exp(tK)$ for some skew-Hermitian matrix $K$. Here
$t\mapsto U_t$ is a path in the unitary matrix space.

We need an integral representation of operator monotone function:
\begin{eqnarray*}
a^r = \frac{\sin (r\pi)}{\pi}\int^{+\infty}_0
\frac{a}{a+x}x^{r-1}dx,
\end{eqnarray*}
where $0<r<1$, $a>0$. For convenience, let $\mu(x) = x^r$. Then we
have
\begin{eqnarray*}
a^r = \frac{\sin (r\pi)}{r\pi}\int^{+\infty}_0 \frac{a}{a+x}d\mu(x),
\end{eqnarray*}
where $r\in(0,1)$, $a\in(0,+\infty)$.

Assume that all the operations are taken on the support of
operators. Given a nonnegative operator $A$, we have:
\begin{eqnarray*}
A^r = \frac{\sin (r\pi)}{r\pi}\int^{+\infty}_0
A(A+x)^{-1}d\mu(x),~~~r\in(0,1).
\end{eqnarray*}
In particular, for $r=\tfrac12$, we have
\begin{eqnarray}
\sqrt{A} = \frac 2\pi\int^{+\infty}_0 A(A+x)^{-1}d\mu(x),
\end{eqnarray}
which gives rise to
\begin{eqnarray*}
\frac{d\sqrt{A}}{dt} &=& \frac 2\pi\int^{+\infty}_0
\Br{\frac{dA}{dt}(A+x)^{-1} + A \frac{d(A+x)^{-1}}{dt}}d\mu(x)\\[2mm]
&=& \frac 2\pi\int^{+\infty}_0 (A+x)^{-1}\frac{dA}{dt}(A+x)^{-1}
xd\mu(x),
\end{eqnarray*}
and
\begin{eqnarray*}\label{322}
\Tr{\frac{d\sqrt{A}}{dt}} &=& \frac 2\pi\int^{+\infty}_0
\Tr{(A+x)^{-2}\frac{dA}{dt}} x\,d\mu(x)\\
&=& \frac 2\pi
\Tr{\Br{\int^{+\infty}_0(A+x)^{-2}xd\mu(x)}\frac{dA}{dt}}\\
&=& \frac 2\pi \Tr{\varphi(A)\frac{dA}{dt}},
\end{eqnarray*}
where
\begin{eqnarray}
\varphi(A) := \int^{+\infty}_0(A+x)^{-2}xd\mu(x) = \frac\pi4A^{-1/2}.
\end{eqnarray}

Set $A_t = \sqrt{\rho}U_t\sigma U^\dagger_t\sqrt{\rho}$. One has
\begin{eqnarray}
\frac{dA_t}{dt}= \sqrt{\rho}U_t[K,\sigma] U^\dagger_t\sqrt{\rho}.
\end{eqnarray}
Replacing $A$ with $A_t$ in \eqref{322}, we get
\begin{eqnarray*}
\frac{dg(t)}{dt} &=& \frac{d\Tr{\sqrt{A_t}}}{dt} =
\Tr{\frac{d\sqrt{A_t}}{dt}}\\
&=& \frac12 \Tr{A^{-1/2}_t\frac{dA_t}{dt}} \\
&=& \frac12
\Tr{U^\dagger_t\sqrt{\rho}A^{-1/2}_t\sqrt{\rho}U_t[K,\sigma]}.
\end{eqnarray*}
From Theorem~\ref{th:1}, we see that the maximal and minimal values
of $\rF(\rho,U\sigma U^\dagger)$ are attained for some $U$'s
implementing $[\rho, U\sigma U^\dagger]=0$, respectively. Without
loss of generality, we assume that $[\rho,\sigma]=0$. Hence the fact
that $g(t) = \rF(\rho, U_t\sigma U^\dagger_t)$, where
$U_t=\exp(tK)$, achieves its extremum at $t=0$ is characterized by
requiring that
\begin{eqnarray}
0 = \frac{d\,g(t)}{dt}|_{t=0} &=&\frac12
\Tr{K[\sigma,\sqrt{\rho}A^{-1/2}_0\sqrt{\rho}]}
\end{eqnarray}
for all skew-Hermitian matrices $K$. Thus
$[\sigma,\sqrt{\rho}A^{-1/2}_0\sqrt{\rho}]=0$, which is compatible
with $[\rho,\sigma]=0$.

The above discussion also indicates that the real function $g(t)$ is
differentiable at each point over $\real$ for all skew-Hermitian
$K$. That is, $g(t)$ is a continuous function because the unitary
matrix group is path-connected. Therefore
$$
g(\real)=\Br{\rF(\lambda^\downarrow(\rho),\lambda^\uparrow(\sigma)),
\rF(\lambda^\downarrow(\rho),\lambda^\downarrow(\sigma))},
$$
where we are implicitly taking the union over all the images of $g$
for all skew-hermitian $K$. And the set $\Set{\rF(\rho, U\sigma
U^\dagger): U\in \unitary{\cH_d}}$ is identical to the interval
$\Br{\rF(\lambda^\downarrow(\rho),\lambda^\uparrow(\sigma)),
\rF(\lambda^\downarrow(\rho),\lambda^\downarrow(\sigma))}$.
\end{proof}

\begin{remark}
A quantum system usually evolves unitarily with
$\set{U_t=\exp(-\mathrm{i}tH): t\in \real}$ according to certain
Hamiltonian $H$, rather than the whole unitary group. The problem is
then reduced to determine the optimized values: $\min _{t\in\real}
\rF(\rho, U_t\sigma U_t^\dagger)$ and $\max_{t\in\real} \rF(\rho,
U_t\sigma U_t^\dagger)$ for given density operators $\rho$ and
$\sigma$. Note that every matrix Lie group is a smooth manifold.
Thus the unitary matrix group $\unitary{\cH_d}$, a compact group, is
connected if and only if it is path-connected \cite{Baker}. It is
seen that any unitary matrix is path-connected with $\I_d$ via a
path $U_t =\exp(tK)$ for some skew-Hermitian matrix $K$, i.e.
$K^\dagger = - K$. Indeed since any unitary matrix $U$ can be
parameterized in this way for both unitary matrix $U$ and $V$, there
exists a skew-Hermitian matrix $K$ such that $UV^{-1} = \exp(K)$.
Let $U_t = \exp(tK)V$. Then $U_0 = V$ and $U_t = U$. That is, $U_t,
t\in[0,1]$ is a path between $U$ and $V$.

Hence if $[H,\rho]=0$ or $[H,\sigma]=0$, then
\begin{eqnarray*}
\max_{t\in\real} \rF(\rho,U_t\sigma U^\dagger_t) = \min_{t\in\real}
\rF(\rho,U_t\sigma U^\dagger_t) = \rF(\rho,\sigma).
\end{eqnarray*}
Assume that $[H,\rho]\neq0$ and $[H,\sigma]\neq0$, and denote
\begin{eqnarray}
g(t) := \rF(\rho, U_t\sigma U_t^\dagger).
\end{eqnarray}
Clearly since $g(t)$ is a continuous function and the unitary group
$\unitary{\cH_d}$ is compact, the extreme values of
$g(t)$ over $\real$ do exist. Thus the range of $g(t)$ is a closed
interval. But determining the extreme values is very complicated and
difficult. We leave this open question in the future research.
\end{remark}

\section{Relative entropy between unitary orbits}\label{sect:relative-entropy}

We have studied the quantum fidelity between unitary orbits. One may
also consider other measures of 'distance' instead of quantum
fidelity. In this section we consider the relative entropy between
unitary orbits of two quantum states. We first give a Lemma about
vectors and stochastic matrices.

For a given $d$-dimensional real vector $u=[u_1, u_2, \cdots,
u_d]^\t\in\real^d$, we denote
$$
u^\downarrow = \br{u^\downarrow_1,
u^\downarrow_2,\ldots, u^\downarrow_d}^\t
$$
the rearrangement of $u$ in decreasing order, $\set{u^\downarrow_i}$
is a permutation of $\set{u_i}$ and $u^\downarrow_1\geqslant
u^\downarrow_2\geqslant \cdots \geqslant u^\downarrow_d$. Similarly,
we denote
$$
u^\uparrow = \br{u^\uparrow_1, u^\uparrow_2, \cdots,
u^\uparrow_d}^\t
$$
the rearrangement of $u$ in increasing order. A real vector $u$ is
\emph{majorized} by $v$, $u\prec v$, if $\sum^k_{i=1} u^\downarrow_i
\leqslant \sum^k_{i=1} v^\downarrow_i$ for each $k=1,\ldots,d$ and
$\sum^d_{i=1} u^\downarrow_i = \sum^d_{i=1} v^\downarrow_i$. A
matrix $B=[b_{ij}]$ is called \emph{bi-stochastic} if
$b_{ij}\geqslant0$, $\sum^d_{i=1} b_{ij}=\sum^d_{j=1}b_{ij}=1$
\cite{Ando}. A real vector $u$ is majorized by $v$ if and only if
$u= Bv$ for some $d\times d$ bi-stochastic matrix $B$ \cite{Hardy}.

Denote by $\bB_d$ the set of all $d\times d$ bi-stochastic matrices.
A \emph{unistochastic} matrix $D$ is a bi-stochastic matrix
satisfying $D = U\circ \overline{U}$, where $\circ$ is the
\emph{Schur product}, defined between two matrices as $A\circ B =
[a_{ij}b_{ij}]$ for $A=[a_{ij}]$ and $B=[b_{ij}]$; $U$ is a unitary
matrix and $\overline{U}$ is the complex conjugation of $U$. We
denote $\bB^u_d$ the set of all $d\times d$ unistochastic matrices.

Let $\cS_d$ be the permutation group on the set $\set{1,2,\cdots,
d}$. For each $\pi\in \cS_d$, we define a $d\times d$ matrix
$P_\pi=[\delta _{i\pi(j)}]$, $P_\pi u = [u_{\pi(1)},\cdots,
u_{\pi(d)}]^\t$. $P_\pi$ is bi-stochastic and the set of
bi-stochastic matrices is a convex set. The \emph{Birkhoff-von
Neumann theorem} states that the bi-stochastic matrices are given by
the convex hull of the permutation matrices \cite{Watrous}: A
$d\times d$ real matrix $B$ is bi-stochastic if and only if
there exists a probability distribution $\lambda$ on $\cS_d$ such
that $B = \sum_{\pi\in \cS_d} \lambda_\pi P_\pi$.

\begin{lem}
For any two real vectors $u,v\in\real^d$, we have
\begin{eqnarray}
\Set{\inner{u}{Bv}: B\in\bB^u_d} = \set{\inner{u}{Bv}: B\in\bB_d},
\end{eqnarray}
which in turn is identical to the interval
$\Br{\inner{u^\downarrow}{v^\uparrow},
\inner{u^\downarrow}{v^\downarrow}}$.
\end{lem}

\begin{proof}
Firstly, we show that
\begin{eqnarray}\label{first}
\set{\inner{u}{Bv}: B\in\bB_d} =
\br{\inner{u^\downarrow}{v^\uparrow},
\inner{u^\downarrow}{v^\downarrow}}.
\end{eqnarray}
From the Birkhoff-von Neumann theorem, we see that each $B\in\bB_d$
can be written as a convex combination of permutation matrices:
$$
B = \sum_{\pi\in\cS_d}\lambda_\pi P_\pi,\quad\forall\pi\in\cS_d:
\lambda_\pi\geqslant0, \sum_{\pi\in\cS_d}\lambda_\pi = 1.
$$
Thus $\inner{u}{Bv} =
\sum_{\pi\in\cS_d}\lambda_\pi\inner{u^{\downarrow}}{P_\pi
v^{\downarrow}}$. Since for any real numbers $x_1\leqslant \cdots\leqslant x_d$ and
$y_1\leqslant \cdots\leqslant y_d$, one has
$$
\sum _{i=1}^d x_i
y_{d+1-i}\leqslant \sum _{i=1}^d x_1y_{\pi(i)}\leqslant \sum_{i=1}^d
x_iy_i
$$
under any permutation $\pi$, it is seen that
\begin{eqnarray}
\inner{u^{\downarrow}}{v^{\uparrow}}\leqslant\inner{u^{\downarrow}}{P_\pi
v^{\downarrow}}\leqslant\inner{u^{\downarrow}}{v^{\downarrow}},\quad
\forall \pi\in \cS_d.
\end{eqnarray}

As the set $\Set{\inner{u^{\downarrow}}{P_\pi v^{\downarrow}}:
\pi\in\cS_d}$ is discrete and finite, it follows that the convex
hull of this set is a one-dimensional simplex with their boundary
points $\inner{u^{\downarrow}}{v^{\uparrow}}$ and
$\inner{u^{\downarrow}}{v^{\downarrow}}$. Therefore \eqref{first}
holds.

Secondly, we show that $\set{\inner{u}{Bv}: B\in\bB^u_d} = \set{\inner{u}{Bv}: B\in\bB_d}$.
Indeed, since $\bB^u_d$ is a proper subset of $\bB_d$, one has
$\set{\inner{u}{Bv}: B\in\bB^u_d} \subset \set{\inner{u}{Bv}:
B\in\bB_d}$. Now for arbitrary $D\in\bB_d$, clearly $Dv\prec v$, there exists a
unistochastic matrices $D'\in\bB^u_d$ such that $Dv = D'v$
\cite[Thm.11.2.]{Watrous}. This implies that
$\inner{u}{Dv}=\inner{u}{D'v}$ in $\Set{\Inner{u}{Dv}:
D\in\bB^u_d}$. That is
$$\set{\inner{u}{Bv}: B\in\bB^u_d} \supset
\set{\inner{u}{Bv}: B\in\bB_d}.$$ Finally, they are identically to
an interval
$\Br{\Inner{u^\downarrow}{v^\uparrow},\Inner{u^\downarrow}{v^\downarrow}}$.
\end{proof}

In fact, the following consequence can be derived directly from
\eqref{l33} and the Lemma above,
\begin{eqnarray}
\inner{\lambda^\downarrow(A)}{\lambda^\uparrow(B)}\leqslant
\Tr{AUBU^\dagger} \leqslant
\inner{\lambda^\downarrow(A)}{\lambda^\downarrow(B)}
\end{eqnarray}
for arbitrary $U\in \unitary{\cH_d}$. Moreover, since
\begin{eqnarray*}
\Tr{AUBU^\dagger} &=&
\sum_{i,j}\lambda^\downarrow_i(A)\lambda^\downarrow_j(B)\abs{\Innerm{a_i}{U}{b_j}}^2= \inner{\lambda^\downarrow(A)}{D_U\lambda^\downarrow(B)},
\end{eqnarray*}
where $D_U = \Br{\abs{\Innerm{a_i}{U}{b_j}}^2} \in \bB^u_d$, one has
\begin{eqnarray}
\Set{\Tr{AUBU^\dagger}: U\in \unitary{\cH_d}} &=&
\Set{\inner{\lambda^\downarrow(A)}{D_U\lambda^\downarrow(B)}:
D_U\in\bB^u_d} \\
&=& \Br{\inner{\lambda^\downarrow(A)}{\lambda^\uparrow(B)},
\inner{\lambda^\downarrow(A)}{\lambda^\downarrow(B)}}.
\end{eqnarray}
Therefore the set $\Set{\Tr{AUBU^\dagger}: U\in \unitary{\cH_d}}$ is
an interval:
$$
\begin{array}{l}
\Set{\Tr{AUBU^\dagger}: U\in \unitary{\cH_d}}= \Br{
\inner{\lambda^\downarrow(A)}{\lambda^\uparrow(B)},
\inner{\lambda^\downarrow(A)}{\lambda^\downarrow(B)}}.
\end{array}
$$

The relative entropy of two quantum states $\rho$ and $\sigma$ is
defined by
\begin{eqnarray}
\rS(\rho ||\sigma ) &=& \Tr{\rho (\log\rho - \log\sigma)}
\end{eqnarray}
if $\supp(\rho)\subseteq\supp(\sigma)$, where $\supp(\rho)$ is the support
of $\rho$ defined as the span of the eigenvectors
with the corresponding eigenvalues great than zero. Let
$\rH(p||q)$ denote the classical relative entropy of two probability
distributions $p=\set{p_j}$ and $q=\set{q_j}$,
$$
\rH(p||q) =
\begin{cases}
\displaystyle\sum_j p_j(\log p_j - \log q_j),& \mathrm{if}\ \supp(p)\subseteq\supp(q),\\[2mm]
+\infty,& \mathrm{otherwise}.
\end{cases}
$$
Since $\rS(U\rho U^\dagger||\sigma) = -\rS(\rho) -\Tr{U\rho
U^\dagger \log\sigma}$, from \eqref{l33} and the analysis above, we
have the following results for relative entropy:

\begin{thrm}\label{th:3}
For arbitrary given two quantum states $\rho,\sigma\in
\density{\cH_d}$, with $\sigma$ full-ranked,
\begin{eqnarray}
\min_{U\in \unitary{\cH_d}} \rS(U\rho U^\dagger||\sigma) &=&
\rH(\lambda^{\downarrow}(\rho)||\lambda^{\downarrow}(\sigma)),\\
\max_{U\in \unitary{\cH_d}} \rS(U\rho U^\dagger||\sigma) &=&
\rH(\lambda^{\downarrow}(\rho)||\lambda^{\uparrow}(\sigma)).
\end{eqnarray}
Moreover, the set $\Set{\rS(U\rho U^\dagger||\sigma): U\in
\unitary{\cH_d}}$ is an interval,
$$
\begin{array}{l}
\Set{\rS(U\rho U^\dagger||\sigma): U\in
\unitary{\cH_d}}=\Br{\rH(\lambda^\downarrow(\rho)||\lambda^\downarrow(\sigma)),
\rH(\lambda^\downarrow(\rho)||\lambda^\uparrow(\sigma))}.
\end{array}
$$
\end{thrm}

Theorem \ref{th:3} shows that the maximal and minimal values of the relative
entropy between the unitary orbits of two quantum states are determined by
the classical relative entropy of probability distributions
given by the eigenvalues of two density matrices.
In addition, one can show that Theorem \ref{th:3} also gives rise to the following inequality:
\begin{eqnarray}
\rH(\lambda^\downarrow(\rho)||\lambda^\downarrow(\sigma))\leqslant
\rS(\rho||\sigma) \leqslant
\rH(\lambda^\downarrow(\rho)||\lambda^\uparrow(\sigma)).
\end{eqnarray}

\section{Discussions}\label{sect:conclusion}

We have solved the problem of evaluating the fidelity between
unitary orbits of quantum states. The analytical formulas for the
minimal and maximal values have been obtained. It has been also
proved that the fidelity goes through the whole interval between the
minimal and the maximal values.

As a "measure of the distance" between the fixed state and the
evolved one, we have used the fidelity $\rF(\rho,\sigma(t))$, where
$\sigma(t) = e^{-\mathrm{i}tH}\sigma e^{\mathrm{i}tH}$. The analysis
can be also analogously used for other kinds of measures, for
instance, the constrained optimization problem for the relative
entropy:
\begin{eqnarray}
\max_{t\in\real} \rS(U_t\rho
U^\dagger_t||\sigma)\quad\text{and}\quad \min_{t\in\real}
\rS(U_t\rho U^\dagger_t||\sigma),
\end{eqnarray}
where $U_t = e^{-\mathrm{i}tH}$ is the unitary dynamics generated by
a Hamiltonian $H$. The above constrained optimization problems are
related with the speed of quantum dynamical evolution
\cite{Giovannetti,Taddei}.

Our results can be also applied to other subjects in quantum
computation and quantum information processing, such as optimal
quantum control, in which the state $\rho(0)$ at time zero evolves
into the state $\rho(t)$ at time $t$, $\rho(t)=U(t)\rho(0)U^\dag(t)$
for some unitary operator $U(t)$. The unitary operator $U(t)$ is
determined by the Hamiltonian of the system $H(t)$ satisfying the
time-dependent Schr\"odinger equation, $\dot{U}(t)=-iH(t)U(t)$, with
$U(0)=\I$ the identity operator. $H(t)$ is a Hermitian matrix of the
form, $H(t)=H_d+\sum_{i=1}^m v_i(t)H_i$, where $H_d$ is called the
drift Hamiltonian which is internal to the system, and $\sum_{i=1}^m
v_i(t)\,H_i$ is the control Hamiltonian such that the coefficients
$v_i(t)$ can be externally manipulated \cite{NRS,lyf1}. If the target
state is not in the scope of the states that can be generated be the
given Hamiltonian. Then one has to find a "best" unitary operator to
reach a final state such that the best fidelity between the target
state and the final state is attained.

The results obtained in this context can be also used to study the
modified version of super-additivity of relative entropy and that of
sub-multiplicativity of fidelity in \cite{Zhang2013}. In fact, the
concerned problems have a surprisingly rich mathematical structure and
need to be investigated further.

\subsubsection*{Acknowledgements} Zhang would like to thank Shunlong Luo
and Hai-Jiang Yu for useful discussions and comments. The
first-named author is funded by NSFC (No.11301124) and HDU (KYS075612038).
Fei is supported by the NSFC (No.11275131).

\end{document}

%% file: preamble.tex
\usepackage[T1]{fontenc}
\usepackage[sc]{mathpazo}
\usepackage{amsmath}
\usepackage{amssymb}
\usepackage{enumerate}
\usepackage{amsthm}
\usepackage{amsfonts,mathrsfs}

\usepackage{hyperref}
\hypersetup{pdfpagemode=UseNone}

\newcommand{\tinyspace}{\mspace{1mu}}

\newcommand{\abs}[1]{\left\lvert\tinyspace #1 \tinyspace\right\rvert}

\renewcommand{\t}{{\scriptscriptstyle\mathsf{T}}}

\newcommand{\setft}[1]{\mathrm{#1}}

\newcommand{\density}[1]{\setft{D}\left(#1\right)}
\newcommand{\unitary}[1]{\setft{U}\left(#1\right)}

\newcommand{\supp}{{\operatorname{supp}}}

\def\real{\mathbb{R}}

\def\I{\mathbb{1}}

\newenvironment{mylist}[1]{\begin{list}{}{
    \setlength{\leftmargin}{#1}
    \setlength{\rightmargin}{0mm}
    \setlength{\labelsep}{2mm}
    \setlength{\labelwidth}{8mm}
    \setlength{\itemsep}{0mm}}}
    {\end{list}}

\newcommand{\inner}[2]{\langle #1 , #2\rangle}

\newcommand{\out}[2]{| #1\rangle\langle #2 |}
\newcommand{\Inner}[2]{\left\langle #1 , #2\right\rangle}
\newcommand{\Innerm}[3]{\left\langle #1 \left| #2 \right| #3 \right\rangle}

\newcommand{\Pa}[1]{\left(#1\right)}
\newcommand{\br}[1]{[#1]}
\newcommand{\Br}[1]{\left[#1\right]}
\newcommand{\set}[1]{\{#1\}}
\newcommand{\Set}[1]{\left\{#1\right\}}

\newcommand{\ket}[1]{|#1\rangle}

\DeclareMathOperator{\trace}{Tr}

\newcommand{\Ptr}[2]{\trace_{#1}\Pa{#2}}

\newcommand{\Tr}[1]{\Ptr{}{#1}}

\def\cH{\mathcal{H}}

\def\cS{\mathcal{S}}
\def\cU{\mathcal{U}}

\def\bB{\mathbf{B}}

\def\rF{\mathrm{F}}\def\rH{\mathrm{H}}

\def\rS{\mathrm{S}}

\newtheorem{thrm}{Theorem}[section]
\newtheorem{lem}[thrm]{Lemma}

\theoremstyle{definition}

\newtheorem{remark}[thrm]{Remark}

\numberwithin{equation}{section}

\newcounter{questionnumber}